\documentclass[11pt,a4paper]{amsart}

\usepackage{amsfonts,amssymb,amsmath,eucal,pinlabel,array,hhline}
\usepackage[all]{xy}
\usepackage{tabulary}
\usepackage{fancyhdr}
\usepackage{a4wide}
\usepackage{calrsfs}

\newtheorem{theorem}{Theorem}[section]

\newtheorem{lemma}[theorem]{Lemma}
\newtheorem{proposition}[theorem]{Proposition}
\newtheorem{corollary}[theorem]{Corollary}
\theoremstyle{definition}\newtheorem{definition}{Definition}
\theoremstyle{remark}\newtheorem{example}{Example}
\theoremstyle{remark}\newtheorem{remark}[theorem]{Remark}

\newcounter{ticklistc}

\newcommand{\ep}{\varepsilon}
\newcommand{\Z}{\mathbb Z}
\newcommand{\R}{\mathbb R}
\newcommand{\C}{\mathbb C}

\newcommand{\EE}{\mathbb E}

\renewcommand{\L}{\mathcal L}

\newcommand{\Hom}{\mathrm{Hom}}

\begin{document}

\title{The critical temperature for the Ising model on planar doubly periodic graphs}

\author{David Cimasoni}
\address{Section de math\'ematiques, 2-4 rue du Li\`evre, 1211 Gen\`eve 4, Switzerland}
\email{David.Cimasoni@unige.ch; Hugo.Duminil@unige.ch}
\author{Hugo Duminil-Copin}


\begin{abstract}
We provide a simple characterization of the critical temperature for the Ising model on an arbitrary planar doubly periodic weighted  graph.
More precisely, the critical inverse temperature $\beta$ for a graph $G$ with coupling constants $(J_e)_{e\in E(G)}$ is obtained as the unique solution of a linear equation in the
variables $(\tanh(\beta J_e))_{e\in E(G)}$. This is achieved by studying the high-temperature expansion of the model using Kac-Ward matrices.
\end{abstract}

\maketitle


\pagestyle{myheadings}
\markboth{D. Cimasoni and H. Duminil-Copin}{The critical temperature for the Ising model}


\section{Introduction}

\subsection{Motivation}

The Ising model is probably one of the most famous models in statistical physics. It was introduced by Lenz in~\cite{Len20} as an attempt to understand Curie's temperature for ferromagnets. It can be defined as follows.
Let $G$ be a finite graph with vertex set $V(G)$ and edge set $E(G)$. A {\em spin configuration} on $G$ is an element $\sigma$ of $\{-1,+1\}^{V(G)}$.
Given a positive edge weight system $J=(J_e)_{e\in E(G)}$ on $G$, the {\em energy} of such a spin configuration $\sigma$ is defined by
\[
\mathcal{H}(\sigma)=-\sum_{e=\{u,v\}\in E(G)}J_e\sigma_u\sigma_v.
\]
Fixing an {\em inverse temperature} $\beta\ge 0$ determines a probability measure on the set $\Omega(G)$ of spin configurations by
\[
\mu_{G,\beta}(\sigma)=\frac{e^{-\beta\mathcal{H}(\sigma)}}{Z_\beta^J(G)},
\]
where the normalization constant
\[
Z_\beta^J(G)=\sum_{\sigma\in\Omega(G)}e^{-\beta\mathcal{H}(\sigma)}
\]
is called the {\em partition function\/} of the {\em Ising model on $G$ with coupling constants $J$\/}.

In this article, we focus on planar locally-finite doubly periodic weighted graphs $(\mathcal G,J)$, i.e. weighted graphs which are invariant under the action of some lattice $\Lambda\simeq\Z\oplus\Z$. In such case, $\mathcal G/\Lambda=:G$ is a finite graph embedded in the torus $\mathbb T^2=\R^2/\Lambda$.
By convention $\mathcal G$ will always denote a doubly periodic graph embedded in the plane, while $G$ will denote a graph embedded in the torus.  Ising probability measures can be constructed on $\mathcal G$ as limits of finite volume probability measures~\cite{McCW}. In particular, the Ising measure at inverse temperature $\beta$ on $\mathcal G$ with $+$ boundary conditions will be denoted by $\mu_{\mathcal G,\beta}^+$. 

We further assume that $\mathcal G$ (or equivalently $G$) is {\em non-degenerate}, i.e. that  the complement of the edges is the union of topological discs. A Peierls argument~\cite{Pei36} and the GKS inequality~\cite{Gri67,KS68} classically imply that the Ising model on $\mathcal G$ exhibits a phase transition at some critical inverse temperature $\beta_c\in(0,\infty)$:
\begin{itemize}
\item for $\beta<\beta_c$, $\mu_{\mathcal G,\beta}^+(\sigma_v)=0$ for any $v\in V(\mathcal G)$,
\item for $\beta>\beta_c$, $\mu_{\mathcal G,\beta}^+(\sigma_v)>0$ for any $v\in V(\mathcal G)$.
\end{itemize}
This article provides a computation of the critical inverse temperature for arbitrary non-degenerate doubly periodic weighted graphs $(\mathcal G,J)$
as a solution of an algebraic equation in the variables $x_e=\tanh(\beta J_e)$.

\subsection{High-temperature expansion of the Ising model}
The result will be best stated in terms of the high-temperature expansion of the Ising model, which we present briefly now. As observed by van der Waerden~\cite{vdW}, the identity
\[
\exp(\beta J_e\sigma_u\sigma_v)=\cosh(\beta J_e)(1+\tanh(\beta J_e)\sigma_u\sigma_v)
\]
allows to express the partition function as
\begin{align}
Z_\beta^J(G)&=\Big(\prod_{e\in E(G)}\cosh(\beta J_e)\Big)\sum_{\sigma\in\Omega(G)}\prod_{e=[uv]\in E(G)}(1+\tanh(\beta J_e)\sigma_u\sigma_v)\nonumber\\
&=\Big(\prod_{e\in E(G)}\cosh(\beta J_e)\Big)2^{|V(G)|}\sum_{\gamma\in\mathcal{E}(G)}\prod_{e\in\gamma}\tanh(\beta J_e),\label{high-temp}
\end{align}
where $\mathcal{E}(G)$ denotes the set of even subgraphs of $G$, that is, the set of subgraphs $\gamma$ of $G$ such that every vertex of $G$ is adjacent to an even number of edges of $\gamma$. As a consequence, the Ising partition function
is proportional to
$
Z(G,x):=\sum_{\gamma\in\mathcal{E}(G)}x(\gamma),
$
where $x(\gamma)=\prod_{e\in\gamma}x_e$ and $x_e=\tanh(\beta J_e)$.
This is called the {\em high-temperature expansion\/} of the partition function.

\subsection{Statement of the result} Let $(\mathcal G,J)$ be a planar non-degenerate locally-finite doubly periodic weighted graph. Recall that $G=\mathcal G/\Lambda$ is naturally embedded in the torus. Let $\mathcal{E}_0(G)$ denote the set of even subgraphs of $G$ that wind around each of the two directions of the torus an even number of times. Set $\mathcal{E}_1(G)=\mathcal{E}(G)\setminus\mathcal{E}_0(G)$. 

\begin{theorem}
\label{thm:temp}
The critical inverse temperature $\beta_c$ for the Ising model on the weighted graph $(\mathcal G,J)$ is the unique solution $0<\beta<\infty$ to the equation
\begin{equation}\label{eq:temp}
\sum_{\gamma\in\mathcal{E}_0(G)}x(\gamma)=\sum_{\gamma\in\mathcal{E}_1(G)}x(\gamma),
\end{equation}
where $x(\gamma)=\prod_{e\in\gamma}x_e$ and $x_e=\tanh(\beta J_e)$. 
\end{theorem}

Note that the sign of $\sum_{\gamma\in\mathcal{E}_0(G)}x(\gamma)-\sum_{\gamma\in\mathcal{E}_1(G)}x(\gamma)$ indicates if the system is in the disordered phase (when the quantity is positive)
or ordered phase (when it is negative).


Let us briefly summarize the strategy of the proof.
Our main tool is the so-called {\em Kac-Ward matrix\/} associated to the weighted graph $(G,J)$ and to a pair of non-vanishing complex numbers $(z,w)$, see Definition~\ref{def:KW} below.
We show that the free energy per fundamental domain can be expressed in terms of the Kac-Ward determinants (Lemma~\ref{lemma:free}). Using exponential decay
of the spin correlations in the disordered phase~\cite{ABF87} together with a duality argument (Theorem~\ref{normal}), the free energy is then shown to be twice differentiable in each variable
$J_e$ at any $\beta\neq \beta_c$. The proof is completed by observing that Equation~\eqref{eq:temp} is equivalent to the vanishing of the Kac-Ward determinant at $(z,w)=(1,1)$, which translates
into the free energy not being twice differentiable in some variable $J_e$.

The main technical difficulty is to make sure that these Kac-Ward determinants never vanish for $z$ and $w$ of modulus 1, except at $(z,w)=(1,1)$. This is achieved by showing that these determinants
are proportional to the Kasteleyn determinants of an associated bipartite graph (Theorem~\ref{thm:corr}), a graph first considered by Fan-Wu~\cite{F-W} (see also~\cite{Dub}). One can then
harness~\cite{K-O} to show that the associated spectral curve is a special Harnack curve, and the desired statement follows (Lemma~\ref{lemma:pos}).

This connection between the Ising model and dimers on a bipartite graph is interesting on its own. It was used recently by Dub\'edat to show the exact bozonisation of the
Ising spin-correlations~\cite{Dub}. We give some other applications in the body of the text, and refer to~\cite{Cim12} for further ones.
Let us also mention that Kac-Ward matrices are related to $s$-holomorphicity and fermionic observables introduced in~\cite{CS09} (see~\cite{DS11} for an overview on the subject). Therefore, this identification of the critical inverse temperature using the Kac-Ward matrices opens new grounds in the understanding of universality and conformal invariance for the Ising model on arbitrary
doubly periodic graphs. 

We would like to highlight the fact that in the case of doubly periodic coupling constants on $\mathcal G=\mathbb Z^2$, Theorem~\ref{thm:temp} was proved independently by Li~\cite{Li2}
(although the statement is formulated quite differently). The proof in~\cite{Li2} relies on the paper~\cite{Li1} of the same author. The main ingredient is a sequence of transformations
mapping the Ising model on $\mathbb Z^2$ to the dimer model on the associated Fisher graph, and ultimately to a dimer model on a bipartite graph using~\cite{Dub}.
In our strategy, the Kac-Ward matrices allow us to work on the original graph almost throughout, avoiding in particular the use of the Fisher correspondence.

\subsection{A corollary} A consequence of the proof is the following corollary, which identifies the critical point as being the unique singularity of the free energy per fundamental domain, see \eqref{eq:def free energy} for the definition.

\begin{corollary}\label{cor}
The free energy per fundamental domain is analytic in $\beta$ for any $\beta \ne \beta_c$.
\end{corollary}

\subsection{Some examples}

We now illustrate Theorem~\ref{thm:temp} with several examples.

\begin{figure}[Htb]
\labellist\small\hair 2.5pt
\pinlabel {$J_1$} at 1590 150
\pinlabel {$J_2$} at 1780 150
\pinlabel {$J_3$} at 1530 45
\pinlabel {$J_4$} at 1760 45
\endlabellist
\centerline{\psfig{file=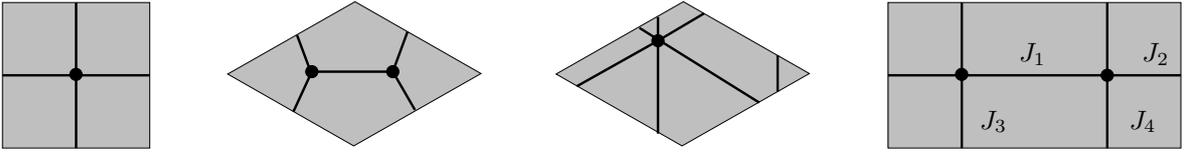,width=\textwidth}}
\caption{The graphs of Examples~\ref{ex:1},~\ref{ex:2},~\ref{ex:3} and~\ref{ex:5}.}
\label{fig:ex}
\end{figure}

\begin{example}\label{ex:1}
Consider the square lattice with coupling constants $J_e=1$. It can be represented by the toric graph illustrated to the left of Figure~\ref{fig:ex}.
The equation of Theorem~\ref{thm:temp} then reads
$1=2x+x^2$, where $x=\tanh(\beta)$. This leads to the well-known value 
\[
\beta_c=\tanh^{-1}(\sqrt{2}-1)=\frac{1}{2}\log(\sqrt{2}+1).
\]
This value was first predicted in~\cite{K-W} and proved in~\cite{Ons}. Several alternative derivations have been presented, see e.g.~\cite{ABF87,BD11,BD12}.
\end{example}

\begin{example}\label{ex:2}
Consider now the hexagonal lattice with coupling constants $J_e=1$ (see Figure~\ref{fig:ex}). This time, the equation reads $1=3x^2$, leading to
\[
\beta_c=\tanh^{-1}(\sqrt{3}/3)=\frac{1}{2}\log(2+\sqrt{3}).
\]
\end{example}

\begin{example}\label{ex:3}
Similarly, the equation associated to the homogeneous triangular lattice as represented in Figure~\ref{fig:ex} is
$
1+x^3=3x+3x^2.
$
This leads to 
\[
\beta_c=\tanh^{-1}(2-\sqrt{3})=\frac{1}{2}\log(\sqrt{3}).
\]
\end{example}

\begin{example}
More generally, consider a doubly periodic graph $\mathcal G$ {\em isoradially embedded} in the plane, i.e. embedded in such a way that each face is inscribed in a circle of radius one,
with the circumcenter in the closure of the face. To each edge $e$, associate the coupling constant
\[
J_e=\frac{1}{2}\log\left(\frac{1+\sin\theta_e}{\cos\theta_e}\right),
\]
where $\theta_e\in(0,\pi/2)$ is the half-rhombus angle associated to the edge $e$, as illustrated in Figure~\ref{fig:theta}.

It follows from~\cite[Theorem~4.7]{Cim3} (see also~\cite{BdT1}) that $\sum_{\gamma\in\mathcal{E}_0(G)}x(\gamma)-\sum_{\gamma\in\mathcal{E}_1(G)}x(\gamma)$ vanishes for $\beta=1$.
By Theorem~\ref{thm:temp}, this is the critical temperature for the isoradial graph $\mathcal G$. The cases with constant angles $\theta_e=\pi/4$, $\pi/3$ and $\pi/6$
correspond to the three examples above.
Let us mention that the class of isoradial graphs has been extensively studied in order to understand universality; see~\cite{Bax78,BdT1,BdT2,CS09}.
\end{example}

\begin{figure}[Htb]
\labellist\small\hair 2.5pt
\pinlabel {$e$} at 180 135
\pinlabel {$\theta_e$} at 115 140
\endlabellist
\centerline{\psfig{file=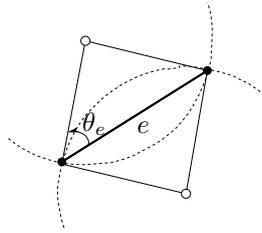,height=3cm}}
\caption{An edge $e$ of an isoradial graph, the associated rhombus, and the half-rhombus angle $\theta_e$.}
\label{fig:theta}
\end{figure}

\begin{example}\label{ex:5} Let us conclude with one last example which does not belong to the class of isoradial weighted graphs.
Let $G$ denote the $2\times 1$ square lattice with arbitrary coupling constants $J_1,J_2,J_3,J_4$ as illustrated in the right-hand side of Figure~\ref{fig:ex}. Then, the critical inverse temperature is
given by the equation
\[
1+x_3x_4=x_3+x_4+x_1x_2+x_1x_2x_3+x_2x_3x_4+x_1x_2x_3x_4
\]
where $x_i=\tanh(\beta J_i)$ for $i=1,2,3,4$.
\end{example}

\subsection*{Organization of the article} The next section defines the Kac-Ward matrices and recalls several of their properties. Section 3 presents a connection between Kac-Ward and Kasteleyn matrices,
as well as consequences for the Ising model. Section 4 contains the proofs of Theorem~\ref{thm:temp} and Corollary~\ref{cor}.

\subsection*{Acknowledgments}The second author was supported by the ANR grant BLAN06-3-134462, the ERC AG CONFRA, as well as by the Swiss
{FNS}. The authors would like to thank Yvan Velenik for very interesting comments and references. The authors also thank Zhongyang Li, B\'eatrice de Tilli\`ere and
C\'edric Boutillier for helpful discussions.


\section{The Kac-Ward matrices}
\label{sec:KW}

The aim of this section is to review the definition and main properties of the Kac-Ward matrices associated with graphs embedded in surfaces. To simplify the present exposition, we shall only
treat the case of toric graphs with straight edges, referring the reader to~\cite{Cim2} for the general case and the proofs.

Let us start with some general terminology and notation. Given a weighted graph $(G,x)$, let ${\EE}={\EE}(G)$ be the set of oriented edges of $G$.
Following~\cite{Ser}, we shall denote by $o(e)$ the origin of an oriented edge $e\in\EE$, by $t(e)$ its terminus, and by $\bar{e}$ the same edge with the opposite orientation.
By abuse of notation, we shall write $x_e=x_{\bar{e}}$ for the weight associated to the unoriented edge corresponding to $e$ and $\bar{e}$.

Now, assume that $G$ is embedded in the oriented torus $\mathbb{T}^2$ so that its edges are straight lines and its faces are topological discs.
Fix a character $\varphi$ of the fundamental group of $\mathbb{T}^2$, that is, an element of
\[
\Hom(\pi_1(\mathbb{T}^2),\C^*)
=H^1(\mathbb{T}^2;\C^*),
\]
the first cohomology group of $\mathbb{T}^2$ with coefficients in $\C^*$.\footnote{Recall that in the present context,
a {\em 1-cochain\/} is a map $\varphi\colon\EE\to\C^*$ such that $\varphi(\bar{e})=\varphi(e)^{-1}$ for all $e\in\EE$.
Such a map is called a {\em 1-cocycle\/} if for each face $f$ of $G\subset\mathbb{T}^2$, $\varphi(\partial f):=\prod_{e\in\partial f}\varphi(e)=1$.
Multiplying each $\varphi(e)$ such that $o(e)=v$ by a fixed $\lambda_v\in\C^*$ results in another 1-cocycle, which is said to be {\em cohomologous\/} to $\varphi$.
Equivalence classes of 1-cocycles define the first cohomology group $H^1(\mathbb{T}^2;\C^*)$.}

Note that the choice of two oriented simple closed curves $\gamma_x,\gamma_y$ representing a basis of $H_1(\mathbb{T}^2;\Z)$ determines an isomorphism
$H^1(\mathbb{T}^2;\C^*)\simeq (\C^*)^2$. Then, such a character simply corresponds to a pair of non-zero complex numbers $(z,w)$. Let us now assume that the curves
$\gamma_x,\gamma_y$ avoid the vertices of $G$ and intersect its edges transversally. Then, a natural 1-cocycle representing the class $(z,w)$ is given by the map
$\varphi\colon\EE\to\C^*$ defined by $\varphi(e)=z^{\gamma_x\cdot e}\,w^{\gamma_y\cdot e}$, where $\cdot$ denotes the intersection form. In words, $\gamma\cdot e$
gathers a $+1$ (resp. $-1$) each time $\gamma$ and $e$ intersect in such a way that $(\gamma,e)$ define the positive (resp. negative) orientation on $\mathbb{T}^2$.

\begin{definition}
\label{def:KW}
Let $T^\varphi$ denote the $|{\EE}|\times|{\EE}|$ matrix defined by
\[
T^\varphi_{e,e'}=
\begin{cases}
\varphi(e)\,\exp\left(\frac{i}{2}\alpha(e,e')\right)\,x_e& \text{if $t(e)=o(e')$ but $e'\neq \bar{e}$;} \\
0 & \text{otherwise,}
\end{cases}
\]
where $\alpha(e,e')\in (-\pi,\pi)$ denotes the angle from $e$ to $e'$, as illustrated in Figure~\ref{fig:alpha}.
We shall call the matrix $I-T^\varphi$ the {\em Kac-Ward matrix\/} associated to the weighted graph $(G,x)$, and denote its determinant by
$P^\varphi(G,x)$. When $\varphi$ is identified with a pair of complex numbers $(z,w)$, we shall simply denote the determinant by $P^{z,w}(G,x)$.
\end{definition}

\begin{figure}[Htb]
\labellist\small\hair 2.5pt
\pinlabel {$e$} at 96 180
\pinlabel {$e'$} at 246 190
\pinlabel {$\alpha(e,e')$} at 310 100
\endlabellist
\centerline{\psfig{file=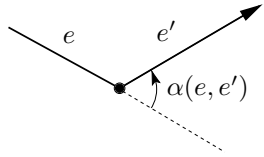,height=2cm}}
\caption{The angle $\alpha(e,e')\in (-\pi,\pi)$.}
\label{fig:alpha}
\end{figure}

Expanding this determinant, see~\cite{Cim2}, leads to an expression of the form
\begin{equation}\label{eq:def p}
P^\varphi(G,x)=\sum_{\gamma\in\Gamma(G)}(-1)^{\epsilon(\gamma)}x(\gamma)\varphi(\gamma),
\end{equation}
where $\epsilon(\gamma)\in\{0,1\}$ is some sign and elements of $\Gamma(G)$ are unions of closed paths in $G$ that never backtrack and pass through
each edge of $G$ at most twice, and if so, in opposite directions.
This shows that $P^\varphi(G,x)$ does not depend on the choice of the 1-cocycle representing the cohomology class $\varphi\in H^1(\mathbb{T}^2;\C^*)$.
Since $\gamma$ belongs to $\Gamma(G)$ if and only if $\bar{\gamma}$ does, we also immediately obtain that if $\varphi$ belongs to
$H^1(\mathbb{T}^2;S^1)\simeq S^1\times S^1$ and $x$ to $\R^{E(G)}$, then $P^\varphi(G,x)$ is a real number.

As detailed in~\cite{Cim2}, a closer analysis of the sign $\epsilon(\gamma)$ leads to the following results.
Fix two oriented simple closed curves $\gamma_x,\gamma_y$ as above, thus identifying $H^1(\mathbb{T}^2;\C^*)$ with $(\C^*)^2$ and $H_1(\mathbb{T}^2;\Z_2)$ with $(\Z_2)^2$.
For $\alpha\in H_1(\mathbb{T}^2;\Z_2)$, let $Z_\alpha$ denote the corresponding partial partition function, that is
\[
Z_\alpha=\sum_{\gamma\in\mathcal{E}(G),[\gamma]=\alpha}x(\gamma),
\]
where the sum ranges over all paths with homology class equal to $\alpha$.
\begin{proposition}[\cite{Cim2}]
\label{prop:Z}
We have the following equalities:
\begin{align*}
P^{1,1}(G,x)&=(Z_{00}-Z_{10}-Z_{01}-Z_{11})^2,\quad\quad\quad P^{-1,1}(G,x)=(Z_{00}+Z_{10}-Z_{01}+Z_{11})^2,\\
P^{1,-1}(G,x)&=(Z_{00}-Z_{10}+Z_{01}+Z_{11})^2,\quad\quad\quad P^{-1,-1}(G,x)=(Z_{00}+Z_{10}+Z_{01}-Z_{11})^2.
\end{align*}

\end{proposition}

In particular, this shows that for $z,w=\pm 1$, $P^{z,w}(G,x)$ is the square of a polynomial in the weight variables $x$.
By construction, the constant coefficient (with respect to the variables $x$) of the polynomial $P^{z,w}(G,x)$ is equal to $1$.
Denoting by $P^{z,w}(G,x)^{1/2}$ the square root of the polynomial with constant coefficient $+1$, we get the following formula.

\begin{theorem}[\cite{Cim2}]
\label{thm:Z}
The Ising partition function for a toric weighted graph $(G,x)$ is given by
\[
Z(G,x)=\frac{1}{2}\Big(-P^{1,1}(G,x)^{1/2}+P^{1,-1}(G,x)^{1/2}+P^{-1,1}(G,x)^{1/2}+P^{-1,-1}(G,x)^{1/2}\Big).
\]
\end{theorem}


\section{Relation to dimers and consequences}

The aim of this section is to show that the Kac-Ward determinant $P^\varphi(G,x)$ is proportional to the
Kasteleyn determinant of an associated bipartite weighted graph $(C_G,y)$ (see Theorem~\ref{thm:corr} below). By Kenyon and Okounkov's~\cite[Theorem 1]{K-O}, it follows that the curve defined
by the zero-set of $P^\varphi(G,x)$ is a (simple) Harnack curve. It also implies some new avatar of Kramers-Wannier duality (Corollary~\ref{cor:KW1}).

Let $(G,x)\subset\mathbb{T}^2$ be a weighted graph, and let us parametrize $x_e\in(0,1)$ by $x_e=\tan(\theta_e/2)$ with $\theta_e\in(0,\pi/2)$.
Following Fan-Wu~\cite{F-W} and Dub\'edat~\cite{Dub}, let us consider the associated weighted graph $(C_G,y)\subset\mathbb{T}^2$ obtained from $G$ as follows.
Replace each edge $e$ of $G$ by a rectangle with the edges parallel to $e$ having weight $\sin(\theta_e)$
while the other two edges have weight $\cos(\theta_e)$. In each corner of each face of $G\subset\mathbb{T}^2$, we now have two vertices; join them with an edge of weight $1$. This is illustrated in
Figure~\ref{fig:C_G}. Note that since the torus is orientable, the graph $C_G$ is bipartite: its vertices can be split into two sets $B\sqcup W$ (say, black and white vertices) such that no edge of
$C_G$ joins two vertices of the same group.

\begin{figure}[Htb]
\labellist\small\hair 2.5pt
\pinlabel {$G$} at 280 370
\pinlabel {$C_G$} at 1200 370
\pinlabel {$\tan(\theta/2)$} at 300 215
\pinlabel {$\sin(\theta)$} at 1225 250
\pinlabel {$\sin(\theta)$} at 1225 115
\pinlabel {$\cos(\theta)$} at 1080 0
\pinlabel {$\cos(\theta)$} at 1300 0
\pinlabel {$1$} at 1055 270
\pinlabel {$1$} at 1055 103
\pinlabel {$1$} at 1375 260
\pinlabel {$1$} at 1375 109
\endlabellist
\centerline{\psfig{file=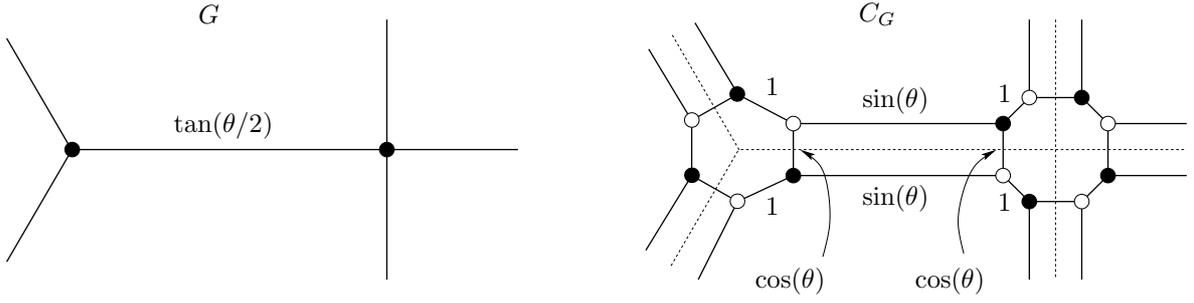,width=\textwidth}}
\caption{The weighted bipartite graph $C_G$ associated to the weighted graph $G$.}
\label{fig:C_G}
\end{figure}

Recall that a {\em Kasteleyn orientation\/}~\cite{Ka1,Ka2,Ka3} on a bipartite graph $C$ embedded in an orientable surface can be understood as a map $\omega\colon E(C)\to\{\pm 1\}$
such that for each face $f$ of $C$,
\[
(\delta\omega)(f):=\prod_{e\in\partial f}\omega(e)=(-1)^{\frac{|\partial f|}{2}+1}.
\]
The associated {\em Kasteleyn operator\/} $K(C,y)\colon\C^B\to\C^W$ can be defined by its matrix elements: for $w\in W$ and $b\in B$, set
\[
K_{wb}=\sum_{e:w\to b}\omega(e)y_e,
\]
the sum being over all edges of $C$ joining $w$ to $b$.
If $C\subset\mathbb{T}^2$ is endowed with a map $\varphi\colon\mathbb{E}(C)\to\C^*$ (for example, a 1-cocycle), then one can extend this definition to an operator $K^\varphi(C,y)\colon\C^B\to\C^W$
by multiplying the coefficient corresponding to the edge $e$ by $\varphi(e)$, where $e$ is oriented from the white to the black vertex.

The main result of this section is the following.

\begin{theorem}
\label{thm:corr}
For any weighted graph $(G,x)\subset\mathbb{T}^2$, there is a Kasteleyn orientation on $C_G\subset\mathbb{T}^2$ such that
\[
P^\varphi(G,x)= 2^{-|V(G)|}\prod_{e\in E(G)}(1+x_e^2)\,\det(K^\varphi(C_G,y))
\]
for all $\varphi\in H^1(\mathbb{T}^2;\C^*)\simeq(\C^*)^2$.
\end{theorem}

\begin{proof}
The strategy of the proof is to gradually transform the Kac-Ward matrix for $(G,x)$ into the Kasteleyn matrix $K^\varphi(C_G,y)$ while keeping track, at each step, of the effect of the transformation
on the determinant. Let us begin with some notation. We shall write $\L({\EE})$ for the complex vector space spanned by the set ${\EE}$ of oriented edges of $G$. Obviously, ${\EE}$ can be partitioned into
${\EE}=\bigsqcup_{v\in V(G)}E_v$, where $E_v$ contains all oriented edges $e$ with origin $o(e)=v$. Now, let us cyclically order the elements of $E_v$ by turning counterclockwise
around $v$. (As $\mathbb{T}^2$ is orientable, this can be done in a consistent way.) Given $e\in E_v$, let $R(e)$ denote the next edge with respect to this cyclic order, as illustrated in
Figure~\ref{fig:R}. This induces an automorphism $R$ of $\L({\EE})$. Also, let $J$ denote the automorphism of $\L({\EE})$ given by $J(e)=\bar{e}$. Finally, we shall write $x$ for the automorphism of
$\L(\EE)$ given by $x(e)=x_e\,e$, and similarly for any weight system and for $\varphi$. 

\begin{figure}[Htb]
\labellist\small\hair 2.5pt
\pinlabel {$e$} at 282 25
\pinlabel {$R(e)$} at 340 265
\pinlabel {$R^2(e)$} at 210 375
\pinlabel {$R^3(e)$} at 70 370
\pinlabel {$\vdots$} at 0 200
\pinlabel {$R^{-1}(e)$} at 100 17
\pinlabel {$\beta_e$} at 220 170
\endlabellist
\centerline{\psfig{file=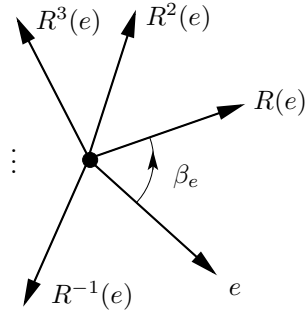,height=4cm}}
\caption{The endomorphism $R$ of $\L({\EE})$, and the angle $\beta_e$.}
\label{fig:R}
\end{figure}

Let $\mathit{Succ}\in\mathit{End}(\L({\EE}))$ be defined as follows: if $e$ is an oriented edge with terminus $t(e)=v$, then
\[
\mathit{Succ}(e)=\varphi(e)x_e\sum_{e'\in E_v}\omega(e,e')\,e',
\]
where $\omega(e,e')=\exp(\frac{i}{2}\alpha(e,e'))$ for $e'\neq \bar{e}\in E_v$ with $\alpha(e,e')$ as in Figure~\ref{fig:alpha}, and $\omega(e,\bar{e})=-i$.
Also, let $T\in\mathit{End}(\L({\EE}))$ be the endomorphism given by $T=\mathit{Succ}+iJ\varphi x$. By definition, the Kac-Ward matrix is equal to
$I-T$. Now, consider the matrix
\[
A=(I-T)(I+iJ\varphi x)=I-\mathit{Succ}+\mathit{Com},
\]
where
\[
\mathit{Com}(e)=-i\varphi(e) x_eT(\bar{e})=-i x_e^2\sum_{\genfrac{}{}{0pt}{}{e'\in E_v}{e'\neq e}}\omega(\bar{e},e')\,e'
\]
if $e$ has origin $o(e)=v$. Since
\[
\det(I+iJ\varphi x)=\prod_{e\in E}\det\left(\begin{array}{cc}1&i\varphi(\bar{e}) x_e\cr i\varphi(e) x_e&1\cr\end{array}\right)=\prod_{e\in E}(1+ x_e^2),
\]
we get the equality
\begin{equation}\label{equ:PA}
P^\varphi(G,x)=\det(I-T)=\prod_{e\in E}(1+x_e^2)^{-1}\det A.
\end{equation}

We now focus on the computation of $\det A$. We shall transform the matrix $A$ by multiplying it to the left with some well-chosen matrix $N$, whose determinant is easy to compute.
Let $Q\in\mathit{End}(\L({\EE}))$ be defined by
\[
Q(e)=\exp\left({\textstyle\frac{i}{2}}\beta_e\right)e,
\]
where $\beta_e=\pi-\alpha(J(R(e)),e)\in(0,2\pi)$ denotes the angle between $e$ and $R(e)$ (see Figure~\ref{fig:R}). Obviously, the endomorphism $N:=I-RQ$ decomposes into
$N=\bigoplus_{v\in V(G)}N_v$ with $N_v\in\mathit{End}(\L({E}_v))$, and one easily computes
\[
\det N_v=1-\prod_{e\in E_v}\exp\left({\textstyle\frac{i}{2}}\beta_e\right)=2.
\]
Hence, the determinant of $N$ is $2^{|V(G)|}$. Now, let us compute the composition $NA$. If $e$ has terminus $t(e)=v$, then
\begin{align*}
N\mathit{Succ}(e)&=(I-RQ)\varphi(e) x_e\sum_{e'\in E_v}\omega(e,e')\,e'\cr
	&=\varphi(e) x_e\sum_{e'\in E_v}\left(\omega(e,e')-\omega(e,R^{-1}(e'))\exp\left({\textstyle\frac{i}{2}}\beta_{R^{-1}(e')}\right)\right)e'\cr
	&=-2i\varphi(e) x_e\,\bar{e}.
\end{align*}
Therefore, we have the equality $N\mathit{Succ}=-2iJ\varphi x$. Similarly, given $e$ with origin $o(e)=v$,
\begin{align*}
N\mathit{Com}(e)=&(I-RQ)(-i)x_e^2\sum_{e'\in E_v\setminus\{e\}}\omega(\bar{e},e')\,e'\cr
	=&-ix_e^2\hskip-.2cm\sum_{e'\in E_v\setminus\{e,R(e)\}}\hskip-.2cm\left(\omega(\bar{e},e')-\omega(\bar{e},R^{-1}(e'))\exp\left({\textstyle\frac{i}{2}}\beta_{R^{-1}(e')}\right)\right)e'\cr
	 &-ix_e^2\left(\omega(\bar{e},R(e))\,R(e)-\omega(\bar{e},R^{-1}(e))\exp\left({\textstyle\frac{i}{2}}\beta_{R^{-1}(e)}\right)\,e\right)\cr
	=&- x^2_e(I+RQ)(e).
\end{align*}
These two equalities lead to
\begin{align*}
NA&=N(I-\mathit{Succ}+\mathit{Com})\cr
	&=(I-RQ)+2iJ\varphi x-(I+RQ) x^2\cr
	&=(1-x^2)+2iJ\varphi x-RQ(1+ x^2).
\end{align*}
Since the determinant of $N$ is $2^{|V(G)|}$, the equality displayed above together with Equation~\eqref{equ:PA} give
\begin{equation}\label{equ:PM}
P^\varphi(G,x)=2^{-|V(G)|}\prod_{e\in E(G)}(1+x_e^2)\det M,
\end{equation}
where $M$ is given by
\[
M=\frac{1-x^2}{1+x^2}+iJ\varphi\,\frac{2x}{1+x^2}-RQ=\cos(\theta)+iJ\varphi\sin(\theta)-RQ,
\]
using the paramatrization $x_e=\tan(\theta_e/2)$ of the weights.

The final step is now to show that this matrix $M$ is conjugate to the Kasteleyn matrix $K^\varphi(C_G,y)$. The theorem will then follow from Equation~\eqref{equ:PM}.
Let $\psi_B\colon\EE\to B$ (resp. $\psi_W\colon\EE\to W$) denote the bijection mapping each oriented edge $e$ of $G$ to the unique black (resp. white) vertex of $C_G$ immediately to the right
(resp. left) of $e$, as illustrated below. Observe that the three maps $\psi_B\circ\psi_W^{-1}$, $\psi_B\circ J\circ \psi_W^{-1}$ and $\psi_B\circ R\circ\psi_W^{-1}$ associate to a fixed $w\in W$ the three black vertices
of $C_G$ adjacent to $w$.

\begin{figure}[h]
\labellist\small\hair 2.5pt
\pinlabel {$e=\psi_W^{-1}(w)$} at 530 15
\pinlabel {$R(e)$} at -40 160
\pinlabel {$\psi_B(R(e))$} at 90 160
\pinlabel {$\psi_B(e)$} at 215 10
\pinlabel {$\psi_B(J(e))$} at 445 115
\pinlabel {$w$} at 157 115
\endlabellist
\centerline{\psfig{file=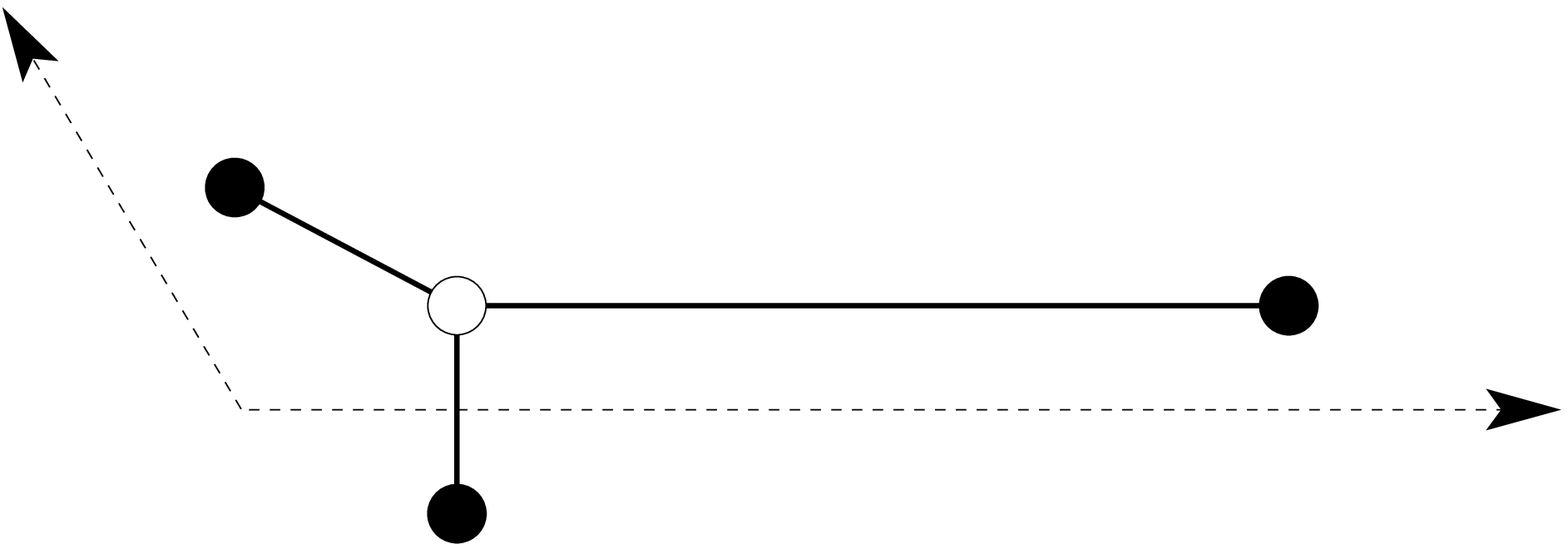,height=2.3cm}}
\end{figure}

\noindent Therefore, the operator
\[
\widetilde{K}^\varphi=(\psi_B\circ M\circ\psi_W^{-1})^*\colon\C^B\to\C^W
\]
is given by the coefficients
\[
\widetilde{K}^\varphi_{wb}=
\begin{cases}
\cos(\theta_e)&\text{if the edge $(w,b)$ is perpendicular to $e\in\EE$;}\\ 
\varphi(e) i \sin(\theta_e)&\text{if $(w,b)$ is to the left of $e\in\EE$;} \\
-\exp\left({\textstyle\frac{i}{2}}\beta_e\right)&\text{if $(w,b)$ is in the ``corner" of $e$ and $R(e)$;} \\
0 & \text{if $w$ and $b$ are not adjacent in $C_G$.}
\end{cases}
\]
This is illustrated below.

\begin{figure}[h]
\labellist\small\hair 2.5pt
\pinlabel {$e$} at 610 75
\pinlabel {$R(e)$} at -40 260
\pinlabel {$\varphi(e) i \sin(\theta_e)$} at 340 140
\pinlabel {$\varphi(e)^{-1} i \sin(\theta_e)$} at 340 10
\pinlabel {$\cos(\theta_e)$} at 95 40
\pinlabel {$\cos(\theta_e)$} at 590 15
\pinlabel {$-\exp\left({\textstyle\frac{i}{2}}\beta_e\right)$} at 310 220
\endlabellist
\centerline{\psfig{file=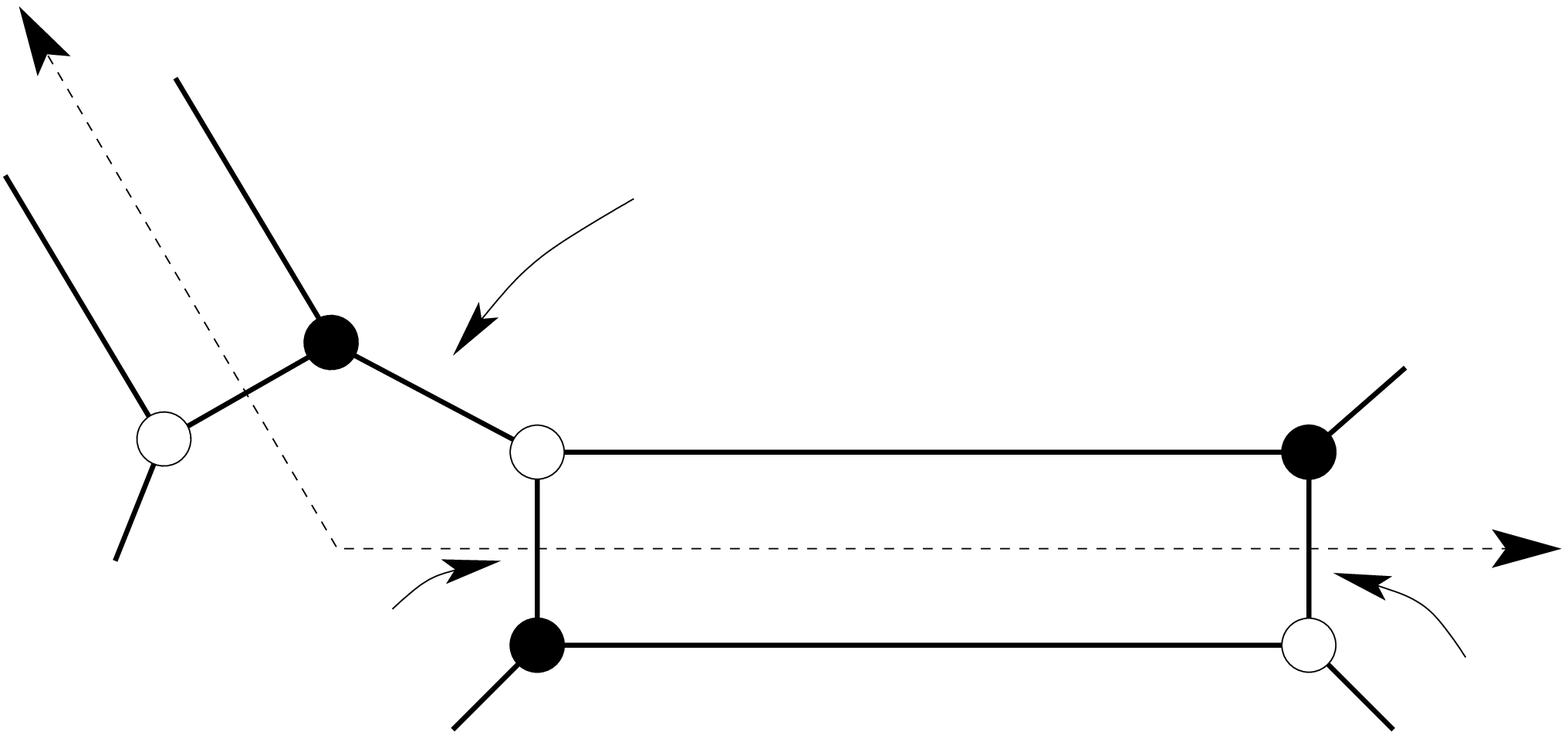,height=3cm}}
\end{figure}

\noindent In other words, $\widetilde{K}^\varphi$ is precisely the Kasteleyn operator of $(C_G,y)$ associated to the 1-cocycle $\tilde\varphi\colon\EE(C_G)\to\C^*$ given by
\[
\tilde\varphi(w,b)=
\begin{cases}
\varphi(e)&\text{if $(w,b)$ runs parallel to $e\in\EE$;}\\ 
1&\text{else,}
\end{cases}
\]
and to the map $\tilde\omega\colon E(C_G)\to S^1\subset\C^*$ given by
\[
\tilde\omega(w,b)=
\begin{cases}
1&\text{if $(w,b)$ is perpendicular to $e\in\EE$;}\\ 
i&\text{if $(w,b)$ is to the left of $e\in\EE$;} \\
-\exp\left({\textstyle\frac{i}{2}}\beta_e\right)&\text{if $(w,b)$ is in the corner of $e$ and $R(e)$.}
\end{cases}
\]
Since the 1-cocycles $\varphi$ and $\tilde\varphi$ induce the same class in $H^1(\mathbb{T}^2;\C^*)$, it only remains to handle the map $\tilde\omega$.
Extend it to a 1-cochain $\tilde\omega\colon\EE(C_G)\to S^1$ by setting $\tilde\omega(b,w)=\tilde\omega(w,b)^{-1}$. Now, observe that for any face
$f$ of $C_G\subset\mathbb{T}^2$,
\[
(\delta\tilde\omega)(f):=\prod_{e\in\partial f}\tilde\omega(e)=(-1)^{\frac{|\partial f|}{2}+1}.
\]
This is obvious for the rectangular faces; for the faces corresponding to vertices of $G$, use the fact that the angles $\beta_e$ add up to $2\pi$ around each vertex;
for the faces corresponding to faces of $G$, use the fact that the angles $\alpha(e,e')$ add up to $2\pi$ around each face. Furthermore, since we are working on the flat torus,
one easily checks that $\tilde\omega(\gamma)=\pm 1$ if $\gamma$ denotes a 1-cycle in $C_G\subset\mathbb{T}^2$. Therefore, $\tilde\omega$ is cohomologous to a Kasteleyn orientation $\omega$.
In other words, $\tilde\omega$ can be transformed into $\omega$ by a sequence of the following transformation: multiply all the edges adjacent to a fixed vertex of $C_G$ by some complex number
of modulus 1. Therefore, a Kasteleyn matrix $K^\varphi(C_G,y)$ can be obtained from $\widetilde{K}^\varphi$ by multiplying lines and columns by complex numbers of modulus 1, and the equality
\[
P^\varphi(G,x)= 2^{-|V(G)|}\prod_{e\in E(G)}(1+x_e^2)\,\det(K^\varphi(C_G,y))
\]
holds up to multiplication by a fixed complex number of modulus 1. For $\varphi$ taking values in $\{\pm 1\}$, both sides are real; therefore, the equality holds up to sign. One can then choose
a Kasteleyn orientation such that the identity holds.
\end{proof}

Let us now fix a geometric basis of $H_1(\mathbb{T}^2;\Z)$. As explained in Section~\ref{sec:KW}, this allows to identify any element $\varphi$ of
$H^1(\mathbb{T}^2;\C^*)$ with a pair of non-zero complex numbers $(z,w)$ and to write $P^\varphi(G,x)=P^{z,w}(G,x)$.
The {\em Newton polygon\/} of such a polynomial $P^{z,w}(G,x)=\sum_{n,m\in\Z}a_{nm}z^nw^m$ is the convex hull of $\{(n,m)\in\Z^2\,;\,a_{nm}\neq 0\}$.
Recall that a weighted graph $(G,x)\subset\mathbb{T}^2$ is {\em non-degenerate\/} if all the weights are positive and the complement of $G$ in the torus consists in topological discs. Using~(\ref{eq:def p}), one easily checks that for such graphs, the Newton polygon of $P^{z,w}(G,x)$ has positive area.
Theorem~\ref{thm:corr} and~\cite[Theorem 1]{K-O} then immediately imply that (the real part of) the associated spectral curve
\[
A=\{(z,w)\in(\C^*)^2\,;\,P^{z,w}(G,x)=0\}\subset(\C^*)^2
\]
is a (possibly singular) {\em Harnack curve\/}. This means that the curve $A\subset(\C^*)^2$ intersects each torus $\mathbb{T}^2(r,s)=\{(z,w)\in(\C^*)^2\,;\,|z|=r,|w|=s\}$
in at most two points (see~\cite{M-R}).

Recall that an element $(z_0,w_0)$ of a complex curve $A=\{(z,w)\in(\C^*)^2\,;\,P(z,w)=0\}$ is called a {\em singularity\/} of $A$ if
$\frac{\partial}{\partial z}P(z,w)=\frac{\partial}{\partial w}P(z,w)=0$. As shown in~\cite[Lemma 6]{M-R}, Harnack curves only admit (a very specific type of) real singularities,
meaning that both coordinates are real. In our case, this translates into the following corollary:

\begin{corollary}
\label{cor:sing}
The spectral curve associated to a non-degenerate weighted graph embedded in the torus has no singularities other than real ones.
\end{corollary}

We now turn to Kramers-Wannier type duality for the Kac-Ward determinants. Again, these results hold for the more general case of an arbitrary weighted graph embedded in a closed orientable surface.
For the simplicity of this exposition, we shall only consider the special case of the torus.

If $G$ is embedded in the torus, its {\em dual\/} is the graph $G^*\subset\mathbb{T}^2$ obtained as follows: each face of $G\subset\mathbb{T}^2$
defines a vertex of $G^*$, and each edge of $G$ bounding two faces of $G\subset\mathbb{T}^2$ defines an edge between the two corresponding vertices of $G^*$. Note that $(G^*)^*=G$. Finally, if $G$ is endowed with weights $x=(x_e)\in(0,1)^{E(G)}$, define the {\em dual weights} $x^*=(x^*_e)\in(0,1)^{E(G)}$ via the condition $x+x^*+xx^*=1$.
If we use the parametrization $x=\tan(\theta/2), x^*=\tan(\theta^*/2)$, then $\theta$ and $\theta^*$ are simply related by $\theta+\theta^*=\pi/2$.
Therefore, the weighted graph $(C_{G^*},y(x^*))$ associated to $(G^*,x^*)$ is equal to the weighted graph $(C_G,y(x))$ associated to $(G,x)$.
Hence, Theorem~\ref{thm:corr} together with the equality $\frac{1+x^2}{1+x}=\frac{1+(x^*)^2}{1+x^*}$ immediately lead to the following.

\begin{corollary}
\label{cor:KW1}
For any toric weighted graph $(G,x)$ and any $\varphi\in H^1(\mathbb{T}^2;\C^*)$,
\[
2^{|V(G)|}\prod_{e\in E(G)}(1+x_e)^{-1} P^\varphi(G,x)=2^{|V(G^*)|}\prod_{e\in E(G)}(1+x^*_e)^{-1} P^\varphi(G^*,x^*).\qed
\]
\end{corollary}

As mentioned in Proposition~\ref{prop:Z}, for $\varphi=(z,w)\in\{\pm 1\}^2$, $P^\varphi(G,x)$ is the square of a polynomial in the weight variables $x_e$.
As the constant coefficient of $P^\varphi(G,x)$ is equal to $1$, we can pick such a square root $P^\varphi(G,x)^{1/2}$ by requiring its constant coefficient to be $+1$.
Taking a closer look at the sign leads to the following duality.

\begin{corollary}
\label{cor:KW2}
For any toric weighted graph $(G,x)$ and any $\varphi\in H^1(\mathbb{T}^2;\{\pm 1\})$,
\[
2^{|V(G)|/2}\prod_{e\in E(G)}(1+x_e)^{-1/2} P^\varphi(G,x)^{1/2}=(-1)^{A(\varphi)}2^{|V(G^*)|/2}\prod_{e\in E(G)}(1+x^*_e)^{-1/2} P^\varphi(G^*,x^*)^{1/2},
\]
where $A(\varphi)=1$ if $\varphi=(1,1)$, and $A(\varphi)=0$ else.
\end{corollary}
\begin{proof}
By Corollary~\ref{cor:KW1}, we only need to determine the sign $A(\varphi)$ in the equation above. Setting $x=1$ (and therefore, $x^*=0$) leads to
\[
P^\varphi(G,1)^{1/2}=(-1)^{A(\varphi)}2^{(|V(G^*)|+|E(G)|-|V(G)|)/2}=(-1)^{A(\varphi)}2^{|V(G^*)|},
\]
using the fact that $|V(G)|-|E(G)|+|V(G^*)|$ is equal to the Euler characteristic of the torus, i.e. zero. Furthermore, the Ising partition function $Z(G,x)$ with weights $x=1$ is nothing but the cardinality of the
$\Z_2$-vector space of 1-cycles modulo 2 in $G$. Since $G$ is connected, the dimension of this space is classically equal to $|E(G)|-|V(G)|+1=|V(G^*)|+1$. Theorem~\ref{thm:Z} now reads
\[
2^{|V(G^*)|+1}=\frac{1}{2}\big(-(-1)^{A(1,1)}+(-1)^{A(1,-1)}+(-1)^{A(-1,1)}+(-1)^{A(-1,-1)}\big)2^{|V(G^*)|}.
\]
The term in parentheses is therefore equal to 4, a fact which determines the sign of the four terms. The corollary follows.\end{proof}


\section{The critical temperature}

Consider a planar non-degenerate locally-finite weighted graph $(\mathcal G,J)$ invariant under a lattice $\Lambda\simeq \Z\oplus\Z$. Alternatively, we will use $(\mathcal G,x)$ when working directly with the high-temperature expansion. Recall that $G=\mathcal G/\Lambda$. For integral positive $n,m$, let $\Lambda_{nm}\simeq n\Z\oplus m\Z$ and let 
$G_{nm}$ denote the toric weighted graph given by $\mathcal G/\Lambda_{nm}$. Note that $G_{11}=G$.

Let us introduce the {\em free energy per fundamental domain} as follows:
\begin{equation}\label{eq:def free energy}
\log Z_x:=\lim_{n\rightarrow \infty} \frac{1}{n^2} \log Z(G_{nn},x).
\end{equation}
This definition is justified by classical super-multiplicative properties of partition functions.
The strategy of the proof of Theorem~\ref{thm:temp} is the following. 
We express the free energy in terms of the Kac-Ward determinant, and we show that for any $x$ such that $P^{z,w}(G,x)$ has a zero on
$\mathbb T^2:=\{(z,w):|z|=1,|w|=1\}$, the free energy $\log Z_{x}$ is not twice differentiable at $x$. We then harness standard arguments on the
Ising model to show that $\log Z_{x}$ is twice differentiable except possibly at criticality. We begin by a classical lemma.

\begin{lemma}
\label{lem:enlarge}
For any $(z,w)\in\mathbb T^2$ and $x\in \mathbb R^{E(G)}$,
\[
P^{z,w}(G_{nm},x)=\prod_{u^n=z}\prod_{v^m=w}P^{u,v}(G,x).
\]
\end{lemma}

\begin{proof}
The proof of~\cite[Theorem 3.3]{KOS} applies almost verbatim.
\end{proof}

The next lemma shows that the only zeros of $P^{z,w}(G,x)$ are localized at $(1,1)$.
\begin{lemma}
\label{lemma:pos}
For any $(z,w)\in \mathbb T^2 \setminus \{(1,1)\}$ and $x\in (0,1)^{E(G)}$, $P^{z,w}(G,x)> 0$.
\end{lemma}

\begin{proof}
The proof is divided into four steps.
\smallbreak
\noindent
{\em \textsc{Step 1.} $P^{z,w}(G,x)>0$ for any $(z,w)\in\{(-1,1),(1,-1),(-1,-1)\}$.}
Proposition~\ref{prop:Z} shows that
\[
Z_{00}(G,x)= \frac{1}{4}\Big(P^{1,1}(G,x)^{1/2}+P^{1,-1}(G,x)^{1/2}+P^{-1,1}(G,x)^{1/2}+P^{-1,-1}(G,x)^{1/2}\Big).
\]
Corollary~\ref{cor:KW2} and Theorem~\ref{thm:Z} then give the equality 
$Z_{00}(G,x)=C\cdot Z(G^*,x^*)$, where 
\[
C:=2^{|V(G^*)|/2-|V(G)|/2-1}\prod_e\left(\frac{1+x_e}{1+x_e^*}\right)^{1/2}.
\]
In the same way, $Z_{10}(G,x)$ can be expressed as a linear combination of $P^{1,1}(G^*,x^*)^{1/2}$, $P^{1,-1}(G^*,x^*)^{1/2}$, $P^{-1,1}(G^*,x^*)^{1/2}$ and $P^{-1,-1}(G^*,x^*)^{1/2}$. Using Proposition~\ref{prop:Z} again, we obtain the equality
\[
Z_{10}(G,x)=C\cdot\left(Z_{00}(G^*,x^*)+Z_{10}(G^*,x^*)-Z_{01}(G^*,x^*)-Z_{11}(G^*,x^*)\right),
\]
which leads to $Z_{10}(G,x)\le C\cdot Z(G^*,x^*)=Z_{00}(G,x)$. This argument can be carried out for any homology class $\alpha$, so
\[
Z_\alpha(G,x)\le Z_{00}(G,x)\quad\text{for any $\alpha\in\{00,01,10,11\}$}.
\]
The assumption that $G$ is non-degenerate implies that all $Z_\alpha(G,x)$'s are strictly positive. The statement now follows from the inequality displayed above, and Proposition~\ref{prop:Z}.

\smallbreak
\noindent{\em \textsc{Step 2.}
$P^{z,w}(G,x)\ne 0$ for any $(z,w)$ such that $z^m=-1$ and $w^n=-1$ for some $(m,n)\in \mathbb N^2$.} This follows immediately from the first step applied to $G_{mn}$ and Lemma~\ref{lem:enlarge}.

\smallbreak
\noindent
\textsc{Step 3.} {\em $P^{z,w}(G,x)\ge 0$ for any $(z,w)\in\mathbb{T}^2$.}
Recall that $P^{z,w}(G,x)$ is real for any $(z,w)\in\mathbb{T}^2$. Let us fix $(z,w)$ with $z^n=-1$ and $w^m=-1$. By the second step,
\[
X=\{x\in(0,1)^{E(G)}\;|\;P^{z,w}(G,x)\ge 0\}=\{x\in(0,1)^{E(G)}\;|\;P^{z,w}(G,x)> 0\}.
\]
By continuity of $x\mapsto P^{z,w}(G,x)$, $X$ is therefore both closed and open. It is also non-empty since $P^{z,w}(G,x)$ tends to $1$ as $x$ tends to 0.
By connexity of $(0,1)^{E(G)}$, $X$ is this whole space. By continuity of $(z,w)\mapsto P^{z,w}(G,x)$, the statement follows.

\smallbreak
\noindent{\em \textsc{Step 4.}
$P^{z,w}(G,x)=0$ implies $(z,w)=(1,1)$.}
Assume that $P^{z,w}(G,x)=0$. Then, $(z,w)$ must be a singularity (i.e. satisfy $\frac{\partial}{\partial z}P^{z,w}(G,x)=\frac{\partial}{\partial w}P^{z,w}(G,x)=0$); otherwise,
$P^{z,w}(G,x)$ would take negative values near $(z,w)$, contradicting the third step. By Corollary~\ref{cor:sing}, $(z,w)\in\mathbb T^2$ is real, i.e.
$(z,w)\in \{(-1,-1),(-1,1),(1,-1),(1,1)\}$. By the first step, $(z,w)$ must be equal to $(1,1)$, and the lemma is proved.
\end{proof}

The free energy can be expressed in terms of $P$ as follows:
\begin{lemma}
\label{lemma:free}
For any $x\in (0,1)^{E(G)}$,
\[
\log Z_x= \frac{1}{2(2\pi i)^2} \int_{\mathbb T^2} \log P^{z,w}(G,x) \frac{dz}{z}\frac{dw}{w}.
\]
\end{lemma}

\begin{proof}
The proof is inspired by the proof of~\cite[Theorem 3.5]{KOS}. First note that $P^{z,w}(G,x)>0$ for any $(z,w)\ne (1,1)$ by Lemma~\ref{lemma:pos}, and that 
\begin{equation*}\label{eq:technical}
\log P^{z,w}(G,x)=O\big[\log (|z-1|+|w-1|)\big],
\end{equation*}
which legitimates the integral on the right-hand side. Lemma~\ref{lem:enlarge} and the bounded convergence theorem imply that for $(\ep,\eta)\in\{(-1,-1),(-1,1),(1,-1)\},$
\begin{align*}\frac1{n^2}\log P^{\ep,\eta}(G_{nn},x)&=\frac1{n^2}\sum_{z^n=\ep}\sum_{w^n=\eta}\log P^{z,w}(G,x)\longrightarrow \frac1{(2\pi i)^2}\int_{\mathbb T^2}\log P^{z,w}(G,x)\frac{dz}{z}\frac{dw}w.\end{align*} 
Now, Proposition~\ref{prop:Z} and Theorem~\ref{thm:Z} imply the inequalities
\[P^{-1,1}(G_{nn},x)\le Z(G_{nn},x)^2\le 9/4\max \{P^{-1,-1}(G_{nn},x),P^{-1,1}(G_{nn},x),P^{1,-1}(G_{nn},x)\},\]
which lead to the claim.
\end{proof}

For the next theorem, let us adopt the terminology of the high-temperature expansion. For $(\mathcal G,x)$ biperiodic, let $\mu_{\mathcal G,x}^+$ be the Ising measure on $\mathcal G$ with edge-weights $x$ and $+$ boundary conditions. Let $x^*$ such that $x+x^*+xx^*=1$ be the dual weights obtained by Kramers-Wannier duality.

\begin{theorem}
\label{normal}
 Let $\mathcal G$ be a non-degenerate locally-finite doubly periodic graph and $r\in V(\mathcal G)$. Then,\begin{enumerate}
 \item[(i)] If $\mu_{\mathcal G,y}^+(\sigma_r)=0$ for any weights $y$ in a neighborhood of $x$, then there exists $c=c(x)>0$ such that $\mu_{\mathcal G,x}^+(\sigma_a\sigma_b)\le \exp(-c|a-b|)$, for any $a,b\in V(\mathcal G)$;
 \item[(ii)] If $\mu_{\mathcal G,y}^+(\sigma_r)>0$ for any weights $y$ in a neighborhood of $x$, there exists $c'=c'(x)>0$ such that $\mu^+_{\mathcal G^*,x^*}(\sigma_u\sigma_v)\le \exp(-c'|u-v|)$, for any $u,v\in V(\mathcal G^*)$.
\end{enumerate}
\end{theorem}

While this theorem is not surprising and follows from very classical ingredients, the proof does not appear in the literature. We therefore recall it here. 

\begin{proof}[Proof of Theorem~\ref{normal}]
Let us prove (i). Choose $\beta$ and $J$ in such a way that $x_e=\tanh(\beta J_e)$. The condition implies that $\beta<\beta_c$ for $(\mathcal G,J)$. We conclude by harnessing~\cite[Theorem 1]{ABF87}. Note that this theorem applies in the very general context of finite range Ising models on $\mathbb Z^d$ with periodic coupling constants. In our case, the model is only biperiodic but as discussed by the authors, the proof extends very easily to this framework.

Let us now deal with (ii). We aim to apply (i) to the dual measures. For this reason, it is sufficient to prove that $\mu_{\mathcal G,y}^+(\sigma_r)>0$ implies $\mu_{\mathcal G^*,y^*}^+(\sigma_u)=0$
or equivalently that $\mu_{\mathcal G^*,y^*}^{\rm free}(\sigma_u\sigma_v)\rightarrow 0$ as $|u-v|\rightarrow \infty$ ($u,v\in \mathcal G^*$), where ``free" refers to free boundary conditions. 
Intuitively, this claim is valid since there cannot be a positive spontaneous magnetization for both the primal and dual Ising models. This is best seen in the context of random-cluster models. We thus harness the Edwards-Sokal coupling~\cite{ES}. 

Let $\phi_{\mathcal G,p,2}^1$ be the random-cluster measure on $\mathcal G$ with cluster-weight 2, edge-weights given by
\[
p_e=\frac{2y_e}{1+y_e}
\]
and wired boundary conditions; see~\cite[Section 4.2]{Gri06}. 
The Edwards-Sokal coupling~\cite[Section 1.4]{Gri06} shows that
$$\phi^1_{\mathcal G,p,2}(r\text{ is connected to infinity})=\mu^+_{\mathcal G,y}(\sigma_r)>0$$ 
which implies the existence of an infinite cluster $\phi_{\mathcal G,p,2}^1$-almost surely. Let $(\phi_{\mathcal G,p,2}^1)^*$ be the dual measure, see~\cite[Section 6.1]{Gri06}. 
Since the random-cluster model satisfies the FKG inequality~\cite[Theorem 3.8]{Gri06}, Corollary 9.4.6 of~\cite{She05} implies that there cannot be coexistence of an infinite cluster and an infinite dual cluster $\phi_{\mathcal G,p,2}^1$-almost surely.
Thus, there is no infinite cluster on $\mathcal G^*$ for the dual random-cluster model $(\phi_{\mathcal G,p,2}^1)^*$-almost surely. Now, $(\phi_{\mathcal G,p,2}^1)^*=\phi^{\rm free}_{\mathcal G^*,p^*,2}$ and $\mu^{\rm free}_{\mathcal G^*,y^*}$ are also coupled via the Edwards-Sokal coupling. We obtain
\[
\mu_{\mathcal G^*,y^*}^{\rm free}(\sigma_u\sigma_v)=(\phi^1_{\mathcal G,p,2})^*(u\text{ connected to }v)\longrightarrow 0
\]
as $|u-v|\rightarrow \infty$ ($u,v\in \mathcal G^*$).
\end{proof}

\begin{remark}
For (ii), one can also invoke (with some modifications) the result of Lebowitz and Pfister~\cite{LP81} together with duality.
\end{remark}

\begin{proof}[Proof of Theorem~\ref{thm:temp}]
Fix $(\mathcal G,J)$. Define $x_\beta=(\tanh(\beta J_e))_e$ where $\beta>0$. First, Proposition~\ref{prop:Z} implies that Equation \eqref{eq:temp}
is equivalent to $P^{1,1}(G,x_\beta)^{1/2}=0$. Note that $P^{1,1}(G,x_\beta)^{1/2}$ tends to $1$ as $\beta$ tends to 0 (by definition), and to $-2^{|V(G^*)|}$
as $\beta$ tends to $\infty$ (by Corollary~\ref{cor:KW2}), so that there exists at least one solution (in $\beta$) to this equation.
We now show that there exists a unique such solution by proving that $P^{1,1}(G,x_{\beta})=0$ implies $\beta=\beta_c$.

We first show that $\beta_c\le \beta$ by assuming that $\beta<\beta_c$, or equivalently $x_{\beta}<x_{\beta_c}$, and by seeking for a contradiction.
Since $P^{1,1}(G,x)$ is not constant on a neighborhood of $x_\beta$ (for instance $P^{1,1}(G,x_{\beta'})>0$ for $\beta'$ close enough to $\beta$),
there exist $e\in E(G)$ and $x_{\beta}\le x< x_{\beta_c}$ such that $P^{1,1}(G,x)=0$ and $x_e\mapsto P^{1,1}(G,x)$ is non constant.
Fix such an edge $e$ and weights $x$, and let $x(t)$ be defined by $x(t)_{e'}=x_{e'}$ if $e'\ne e$ and $x(t)_e=x_e+t$.
By Equation~\eqref{eq:def p}, $t\mapsto P^{1,1}(G,x(t))$ is a polynomial of degree exactly 2.
Furthermore, since $P^{1,1}(G,x(0))=0$, $P^{e^{i\theta},e^{i\eta}}(G,x(t))=P^{e^{-i\theta},e^{-i\eta}}(G,x(t))$ and $P^{e^{i\theta},e^{i\eta}}(G,x)\ge 0$ for any $(\theta,\eta,t)$ in a neighborhood of the origin (Lemma~\ref{lemma:pos}), we obtain the following development near $(0,0,0)$:
$$P^{e^{i\theta},e^{i\eta}}(G,x(t))=(a_{11}t^2+a_{22}\theta^2+a_{33}\eta^2+a_{23}\theta\eta)f(t,\theta,\eta),$$
where $f(t,\theta,\eta)=1+o(|t,\theta,\eta|^2)$ is a non-vanishing analytic function, and the coefficients satisfy $a_{22}a_{33}-\frac{1}{4}a_{23}^2\ge 0$
and $a_{11}>0$.
Now,
$$t\longmapsto\int_{-\pi}^{\pi}\int_{-\pi}^{\pi}\log f(t,\theta,\eta){\rm d}\theta {\rm d}\eta$$
is twice differentiable in $t$, so that $\log Z_x$ is twice-differentiable in $x_e$ at $x$ if and only if 
$$t\longmapsto\int_{-\pi}^{\pi}\int_{-\pi}^{\pi} \log (a_{11}t^2+a_{22}\theta^2+a_{33}\eta^2+a_{23}\theta\eta) {\rm d}\theta {\rm d}\eta$$
is twice differentiable at 0. For $t\ne 0$, the second derivative in $t$ of this function equals
$$\int_{-\pi}^{\pi}\int_{-\pi}^{\pi}\frac{2a_{11}(a_{11}t^2+a_{22}\theta^2+a_{33}\eta^2+a_{23}\theta\eta)-4a_{11}^2t^2}{(a_{11}t^2+a_{22}\theta^2+a_{33}\eta^2+a_{23}\theta\eta)^2}{\rm d}\theta {\rm d}\eta.$$
As $t$ tends to $0$, this integral tends to
$$\int_{-\pi}^{\pi}\int_{-\pi}^{\pi}\frac{2a_{11}}{(a_{22}\theta^2+a_{33}\eta^2+a_{23}\theta\eta)}{\rm d}\theta {\rm d}\eta,$$
that is, to $\infty$, since $a_{11}>0$ and $a_{22}a_{33}-\frac{1}{4}a_{23}^2\ge 0$.
But this is in contradiction with the assumption that $x<x_{\beta_c}$ since in this case, exponential decay implies that $\log Z_x$ is twice differentiable. We now justify this last statement. Fix a representative $\{a,b\}\in E(\mathcal G)$ of the edge $e$. Let $\mathcal G_{nn}$ be a fundamental domain of the action of $\Lambda_{nn}$ on $\mathcal G$. Further assume that $\mathcal G_{nn}$ contains the edge $e$. 
Let $E_n$ (resp. $E$) be the set of translates of $e$ in $E(\mathcal G_{nn})$ (resp. $E(\mathcal G)$). Since the definition of the free energy does not depend on the boundary condition, one has
\[
\log Z_x=\lim_{n\rightarrow \infty}\frac{1}{n^2}\log Z(\mathcal G_{nn},x),
\]
where $Z(\mathcal G_{nn},x)$ is the partition function on $\mathcal G_{nn}$ with free boundary conditions. Set $x_e=\tanh(\beta J'_e)$. The high temperature expansion~\eqref{high-temp} shows that
\[
\log Z_x=\lim_{n\rightarrow \infty}\frac{1}{n^2}\log Z^{J'}_\beta(\mathcal G_{nn})-\sum_{e\in E(G)}\log(\cosh(\beta J'_e))-|V(G)|\cdot\log(2),
\]
where $Z^{J'}_\beta(\mathcal G_{nn})$ is defined in the introduction. Since $J'_e$ depends smoothly on $x_e$, it is sufficient to show that $\lim_{n\rightarrow \infty}\frac{1}{n^2}\log Z^{J'}_\beta(\mathcal G_{nn})$ is twice differentiable with respect to $J'_e$. We obtain
\[
\frac{1}{n^2}\frac{\partial^2}{\partial J'_e\,^2}\log Z^{J'}_\beta(\mathcal G_{nn})
=\beta^2\sum_{\{u,v\}\in E_n}\Big(\mu_{\mathcal G_{nn},x}(\sigma_a\sigma_b\sigma_u\sigma_v)-\mu_{\mathcal G_{nn},x}(\sigma_a\sigma_b)\mu_{\mathcal G_{nn},x}(\sigma_u\sigma_v)\Big),
\]
where $\mu_{\mathcal G_{nn},x}$ is the measure on $\mathcal G_{nn}$ with free boundary conditions. Lebowitz's inequality~\cite[Remark (i), p. 91]{Leb74} then yields
\begin{align*}
|\mu_{\mathcal G_{nn},x}(\sigma_a\sigma_b\sigma_u\sigma_v)-&\mu_{\mathcal G_{nn},x}(\sigma_a\sigma_b)\mu_{\mathcal G_{nn},x}(\sigma_u\sigma_v)|\\
&\le\mu_{\mathcal G_{nn},x}(\sigma_a\sigma_u)\mu_{\mathcal G_{nn},x}(\sigma_b\sigma_v) +\mu_{\mathcal G_{nn},x}(\sigma_a\sigma_v)\mu_{\mathcal G_{nn},x}(\sigma_b\sigma_u).
\end{align*}
The comparison between boundary conditions implies $\mu_{\mathcal G_{nn},x}(\sigma_a\sigma_u)\le \mu_{\mathcal G,x}^+(\sigma_a\sigma_u)$, which, together with Lebowitz's inequality and Property (i) of Theorem~\ref{normal}, leads to
\[
\sum_{\{u,v\}\in E_n}\Big|\mu_{\mathcal G_{nn},x}(\sigma_a\sigma_b\sigma_u\sigma_v)-\mu_{\mathcal G_{nn},x}(\sigma_a\sigma_b)\mu_{\mathcal G_{nn},x}(\sigma_u\sigma_v)\Big|\le
2\sum_{\{u,v\}\in E}\exp(-c(x)|u-a|).
\]
The term on the right-hand side is therefore bounded uniformly in $n$ and in $x$, provided that $x$ takes value in a compact subset of $\{x<x_{\beta_c}\}$. The  bounded convergence theorem then implies that $\log Z_x$ is twice differentiable in $J'_e$, and thus in $x_e$. In conclusion, $x$ cannot be smaller than $x_{\beta_c}$, and we obtain that $\beta\ge \beta_c$.
 \medbreak
Let us conclude the proof by showing that $\beta\le \beta_c$.  Corollary~\ref{cor:KW1} shows that $P^{1,1}(G,x_{\beta})=0$ implies $P^{1,1}(G^*,x_{\beta}^*)=0$. Now if $\beta>\beta_c$, or equivalently $x_{\beta}^*<x_{\beta_c}^*$, one can run the previous argument for the dual model, using Property (ii) of Theorem~\ref{normal} in place of Property (i).
 \end{proof}

\begin{proof}[Proof of Corollary~\ref{cor}]
When $\beta\ne \beta_c$, $P^{z,w}(G,x_\beta)>0$ for any $(z,w)\in \mathbb T^2$ and $\beta\mapsto \log P^{z,w}(G,x_\beta)$ is analytic. By Lemma~\ref{lemma:free}, the free energy is a parameter-dependent integral which is analytic at $\beta\ne \beta_c$.
\end{proof}
\bibliographystyle{plain}

\bibliography{Ising}

\end{document}